\documentclass[conference]{IEEEtran}
\usepackage{graphicx,indentfirst}
\usepackage{booktabs}
\usepackage{indentfirst}
\usepackage{amsmath}
\usepackage{cite}
\usepackage{stfloats}
\usepackage{multicol}
\usepackage{makecell}
\usepackage{multirow}
\usepackage{slashbox}
\usepackage{subfigure}
\usepackage[english]{babel}
\usepackage{amsthm}
\usepackage{bm}
\usepackage{amsmath}
\usepackage{multirow}
\usepackage{paralist}
\usepackage{algorithm, algorithmic}
\usepackage{amssymb}
\usepackage{amssymb}
\usepackage{mathrsfs}
\usepackage{amsfonts}
\usepackage{color}

\newtheorem{theorem}{Theorem}

\theoremstyle{plain}

\IEEEoverridecommandlockouts
\begin{document}
\bibliographystyle{IEEE2}
\title{Optimal Pricing-Based Edge Computing Resource Management in Mobile Blockchain}
\author{Zehui Xiong$^1$, Shaohan Feng$^1$, Dusit Niyato$^1$, Ping Wang$^1$, and Zhu Han$^2$	\\
$^1$School of Computer Science and Engineering, Nanyang Technological University (NTU), Singapore	\\
$^2$Department of Electrical and Computer Engineering, University of Houston, Texas, USA  }
\maketitle
\begin{abstract}
As the core issue of blockchain, the mining requires solving a proof-of-work puzzle, which is resource expensive to implement in mobile devices due to high computing power needed. Thus, the development of blockchain in mobile applications is restricted. In this paper, we consider the edge computing as the network enabler for mobile blockchain. In particular, we study optimal pricing-based edge computing resource management to support mobile blockchain applications where the mining process can be offloaded to an Edge computing Service Provider (ESP). We adopt a two-stage Stackelberg game to jointly maximize the profit of the ESP and the individual utilities of different miners. In Stage~I, the ESP sets the price of edge computing services. In Stage~II, the miners decide on the service demand to purchase based on the observed prices. We apply the backward induction to analyze the sub-game perfect equilibrium in each stage for uniform and discriminatory pricing schemes. Further, the existence and uniqueness of Stackelberg game are validated for both pricing schemes.
At last, the performance evaluation shows that the ESP intends to set the maximum possible value as the optimal price for profit maximization under uniform pricing. In addition, the discriminatory pricing helps the ESP encourage higher total service demand from miners and achieve greater profit correspondingly.
\end{abstract}
\begin{IEEEkeywords}
edge computing, resource management, mobile blockchain, mining, game theory, Variational Inequality.
\end{IEEEkeywords}
\section{Introduction}\label{Sec:Introduction}
In traditional online payments with digital transactions, the consensus is reached through a trusted central authority. Such introduced intermediary increases the cost because the nominal fees need to be deducted and paid. In 2008, a purely peer-to-peer electronic payment concept called \textit{Bitcoin} was proposed that avoids this incurred cost caused by online payments~\cite{Bitcoin}. As one popular digital cryptocurrency, Bitcoin can record all digital transactions in an append-only distributed public ledger called \textit{blockchain}, which is maintained by a group of participants, i.e., \textit{miners}. Since the success of Bitcoin~\cite{coindesk}, the blockchain technologies have generated remarkable public interests through a distributed network without intermediary. 
The key issue of the blockchain is a computational process called \textit{mining}, where the transaction records are appended into the main chain through the solution of the \textit{proof-of-work puzzle}. This proof-of-work puzzle consists of considering a set of transactions that are present in the network, solving a mathematical problem which depends on this set and propagating the result to the blockchain network for this solution to reach consensus. Once all these steps are finished successfully, the set of transactions proposed by the miner forms a \textit{block} that is appended to the current blockchain. The first miner which successfully mines the solution of the puzzle and reach the consensus is considered to be the winner to which a certain reward is offered. This process can be referred as the speed game among the miners with different computing capacities~\cite{pass2017fruitchains, conti2017survey}.


However, blockchain has not been adopted widely in mobile applications~\cite{M-commerce}. This is because blockchain mining needs to solve a proof-of-work puzzle, which is resource expensive to implement in mobile devices due to high computing power needed. Therefore, this motivates us to take a step further to reconsider the mining strategies as well as resource management in mobile environment, thereby opening new opportunities for the development of blockchain in mobile applications. In this paper, we consider the optimal pricing-based resource management in mobile blockchain, where an Edge computing Service Provider (ESP) is introduced to support proof-of-work puzzle offloading~\cite{Zhang2016}. In particular, we analyze two pricing schemes, i.e., uniform pricing in which a uniform unit price is applied to all the miners and discriminatory pricing in which different unit prices are assigned to different miners. 
To the best of our knowledge, this is the first work to investigate the mobile blockchain with resource management using game theory. The main contributions of this work are summarized as follows.
\begin{itemize}
\item We formulate a pricing and service demand problem to analyze the interplay between the ESP and miners. In particular, we adopt the two-stage Stackelberg game to model their interplay by jointly maximizing the profit of the ESP and the individual utilities of different miners for mobile blockchain applications.
\item Through backward induction, we first derive a unique Nash equilibrium point among the miners in the second stage, and then investigate the profit maximization of the ESP in the first stage. The existence and uniqueness of the Stackelberg equilibrium are validated analytically for both pricing schemes.
\item We conduct simulations to evaluate the performance of the proposed pricing-based resource management in mobile blockchain. The results show that the ESP intends to set the maximum possible value as the optimal price for profit maximization under uniform pricing. In addition, the discriminatory pricing helps the ESP encourage more service demand from the miners and achieve greater profit.
\end{itemize}

The rest of the paper is organized as follows. Section~\ref{Sec:Model} describes the system model and formulate the two-stage Stackelberg game. Section~\ref{Sec:Solution} analyzes the optimal service demand of miners as well as the profit maximization of the ESP using backward induction for both uniform and discriminatory pricing schemes. The performance evaluations is presented in Section~\ref{Sec:Simulation}, and Section~\ref{Sec:Conclusion} concludes the paper.

\section{System Model and Game Formulation}\label{Sec:Model}
In this section, we first present the system model of the mobile blockchain under our consideration. Then, we formulate the Stackelberg game setting for the pricing-based edge computing resource management for mobile blockchain applications.
\begin{figure}[t]
\centering
\includegraphics[width=.32\textwidth]{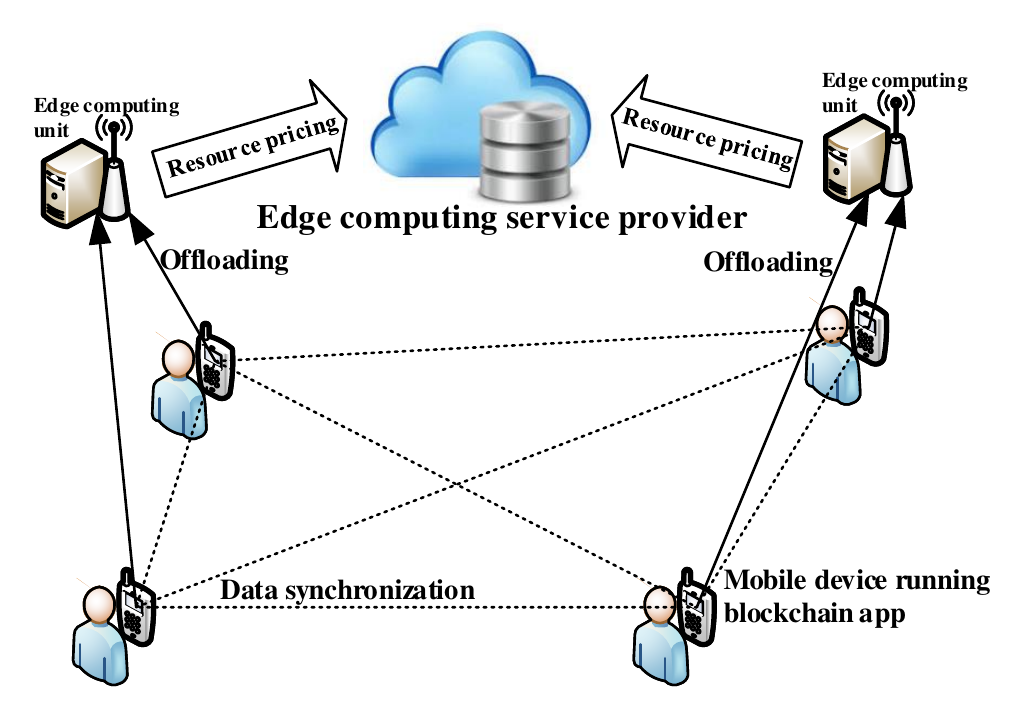}
\caption{System model.}\label{Fig:Model}
\end{figure}
\subsection{Mobile Blockchain Mining with Edge Computing}\label{Subsec:ChainMining}
Figure.~\ref{Fig:Model} illustrates the system model of the mobile blockchain under our consideration. We consider a mobile blockchain application as presented in~\cite{suankaewmanee20xx}, in which there is a group of miners, i.e., the mobile users, ${\mathcal{N}} = \{1,\ldots,N\}$. Each mobile user runs mobile blockchain applications for recording the transactions performed in the group. 
There is an ESP deploying the edge computing units for the miners. The aforementioned proof-of-work puzzle can be offloaded to a nearby edge computing unit. 

The ESP, i.e., the seller, sells the edge computing services, and the miners, i.e., the buyers, access this service from the nearby edge computing unit. Each miner~$i \in {\cal N}$ decides on its service demand, denoted by $x_i$. Additionally, we consider $x_i \in [\underline x ,\overline x]$, in which $\underline x$ is the minimum service demand, e.g., for blockchain data synchronization, and $\overline x$ is the maximum service demand governed by the ESP. Then, let $\mathbf{x} \buildrel \Delta \over = ({x_1}, \ldots ,{x_N})$ and ${\mathbf{x}}_{-i}$ represent the service demand profile of all the miners and all other miners except miner~$i$, respectively. As such, miner $i \in \cal N$ with the service demand $x_i$ has a relative computing power $\alpha_i$ with respect to the total computing power of the network, which is defined as follows:

\begin{equation}\label{Eq:HashingPower}
{\alpha_i}(x_i, {\bf x}_{-i}) = \frac{{{x_i}}}{{\sum\nolimits_{j\in \cal N} {{x_j}} }}, \alpha_i >0,
\end{equation}
such that ${\sum\nolimits_{j\in \cal N} {{\alpha_j}} }=1$.

In the mobile blockchain, miners compete against each other in order to solve the proof-of-work puzzle and receive the mining reward accordingly. The occurrence of solving the puzzle can be modeled as a random variable following a Poisson process with the mean value $\lambda$~\cite{houy2014bitcoin}. 
Once the miner successfully solves the puzzle, the miner needs to propagate its solution to the whole mobile blockchain network and its solution needs to reach consensus. 
The first miner to successfully mine a block that reaches consensus earns the mining reward. The reward is composed of a fixed reward denoted by $R$, and a variable reward which is defined as $r\times t_i$, where $r$ represents a given variable reward factor and $t_i$ represents the number of transactions included in the block mined by miner $i$~\cite{houy2014bitcoin, houy2014economics}. Additionally, the process of solving the puzzle incurs an associated cost, i.e., the payment from miner $i$ to the ESP, $p_i$. The objective of the miners is to maximize their individual expected utility, and for miner $i$, it is given as follows:
\begin{equation}
u_i = (R + r{t_i} )P_i\left(\alpha_i({x_i},{{\bf{x}}_{ - i}}), t_i\right) - {p_i}{x_i},
\end{equation}
where $P\left(\alpha_i ({x_i},{{\bf{x}}_{ - i}}), t_i\right)$ is the probability that miner $i$ successfully mines the block and its solution reaches consensus.

The process of successfully mining a block is composed of two phases, i.e., the mining phase and the propagation phase. In the mining phase, the probability that miner $i$ mines the block is proportional to its relative computing power $\alpha_i$. Furthermore, there are diminishing chances of wining if one miner chooses to propagate a block that propagates slowly to other miners in the propagation phase. In other words, even though one miner may mine the first valid block successfully, if its mined block is large, then this block will be likely to be discarded because of long latency, which is called orphaning~\cite{houy2014bitcoin}. Therefore, the probability of successful mining by miner $i$ is discounted by the chances that the mined block is orphaned, ${\mathbb P}_{\mathrm{orphan}}(t_i)$, which is represented by
\begin{equation}
P_i(\alpha_i ( x_i, {\mathbf{x}}_{-i} ), t_i ) = {\alpha_i}(1 - {{\mathbb P}_{\mathrm{orphan}}}(t_i)).
\end{equation}
Due to the fact that block mining times follow the Poisson distribution aforementioned, the orphaning probability can be approximated as~\cite{Approximation}:
\begin{equation}
{\mathbb P}_{\mathrm{orphan}}(t_i) = 1 - {e^{ - \lambda \tau(t_i) }},
\end{equation}
where $\tau(t_i)$ is the block propagation time, which is a function of the block size. In other words, the propagation time needed for a block to reach consensus is dependent on its size $t_i$, i.e., the number of transactions in it~\cite{decker2013information}. Same as~\cite{houy2014bitcoin}, we assume this time function is linear, i.e., $\tau(t_i) = z \times t_i$ with $z>0$ denotes a given delay factor\footnote{Note that this linear approximation is acceptable according to the numerical results from~\cite{houy2014bitcoin, andrew2015exam}.}. Thus, the probability that miner $i$ successfully mines a block and its solution reaches consensus is expressed as follows:
\begin{equation}
{P_i}(\alpha_i ({x_i},{{\bf{x}}_{ - i}}),t_i) = {\alpha_i}{e^{ -\lambda z{t_i}}},
\end{equation}
where ${\alpha_i}(x_i, {\bf x}_{-i})$ is shown in~(\ref{Eq:HashingPower}).

\subsection{Two-Stage Stackelberg Game Formulation}\label{Subsec:GameModel}
\begin{figure}[t]
\centering
\includegraphics[width=0.4\textwidth]{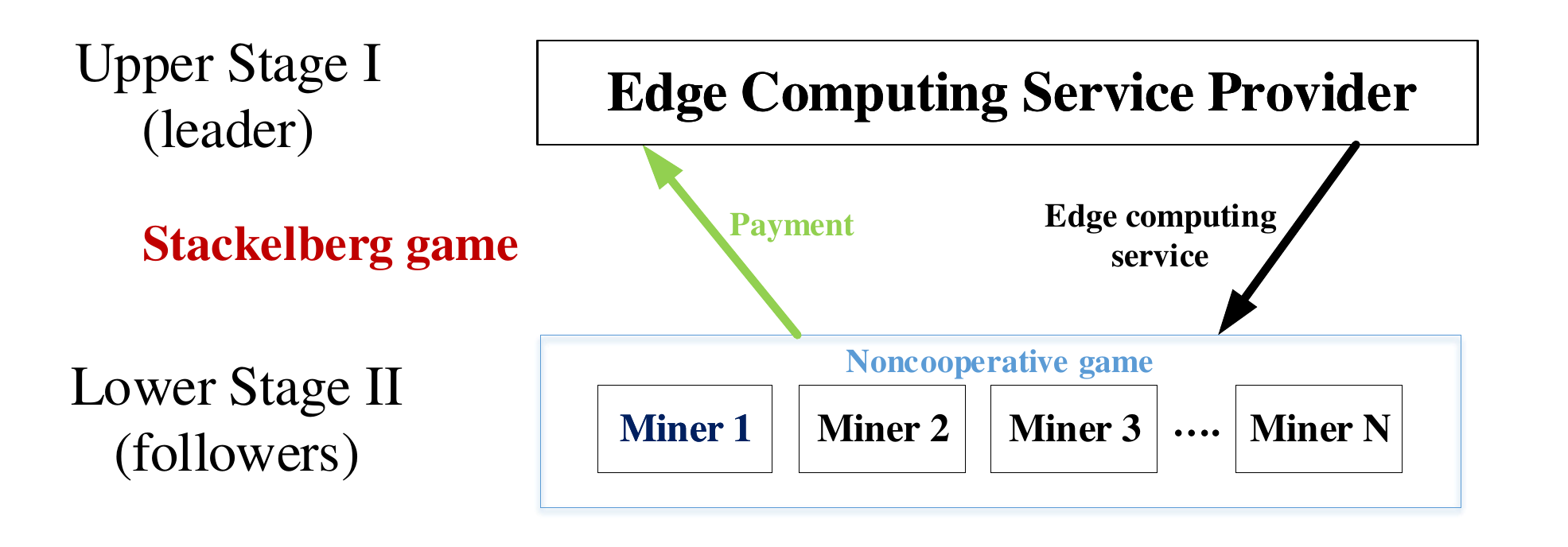}
\caption{Two-stage Stackelberg game model of the interactions among the ESP and miners in the mobile blockchain.}\label{Fig:GameModel}
\end{figure}
The interaction between the ESP and miners is modeled as a two-stage Stackelberg game, as illustrated in Fig.~\ref{Fig:GameModel}. 
\subsubsection{Miners' mining strategies in Stage II}
Given the prices of the ESP and other miners' strategies, miner $i$ decides on its computing service demand to maximize the expected utility which is given as:
\begin{equation}\label{Eq:Utility}
u_i({x_i},{{\bf{x}}_{ - i}},{p_i}) = (R + r{t_i})\frac{{{x_i}}}{{\sum\nolimits_{j\in \cal N} {{x_j}} }}{e^{ - \lambda z{t_i}}} - {p_i}{x_i},
\end{equation}
where $p_i$ is the price per unit for service demand of miner $i$. 

\subsubsection{ESP's pricing strategies in Stage I}
The profit of the ESP is the revenue obtained from charging the miners for computing service minus the service cost. The service cost depends on the service demand $x_i$, the time that the miner takes to mine a block, and the cost of electricity, $c$. Therefore, the ESP determines the prices within the strategy space $\{{\bf p}= [p_i]_{i\in \cal N}:0\le p_i \le \overline p\}$ for maximizing its profit which is defined as:
\begin{equation}\label{Eq:Profit}
\Pi ( {\mathbf{p}}, {\mathbf{x}} ) = \sum\nolimits_{i \in \cal N}{p_i}{x_i} - \sum\nolimits_{i \in \cal N}cT{x_i}.
\end{equation}
%
\begin{figure*}[ht]\scriptsize
\begin{eqnarray}\label{Eq:BestResponse}
{x_i}^* = \mathscr F_i ({{\bf x}}) = \begin{cases}
\underline x, & \sqrt {\frac{{(R + r{t_i})\sum\nolimits_{i \ne j} {{x_j}} }}{{p{e^{ - \lambda z{t_i}}}}}} - \sum\nolimits_{i \ne j} {{x_j}}< \underline x \cr \sqrt {\frac{{(R + r{t_i})\sum\nolimits_{i \ne j} {{x_j}} }}{{p{e^{ - \lambda z{t_i}}}}}} - \sum\nolimits_{i \ne j} {{x_j}}, &\underline x \le \sqrt {\frac{{(R + r{t_i})\sum\nolimits_{i \ne j} {{x_j}} }}{{p{e^{ - \lambda z{t_i}}}}}} - \sum\nolimits_{i \ne j} {{x_j}} \le \overline x \cr \overline x, &\sqrt {\frac{{(R + r{t_i})\sum\nolimits_{i \ne j} {{x_j}} }}{{p{e^{ - \lambda z{t_i}}}}}} - \sum\nolimits_{i \ne j} {{x_j}}> \overline x\end{cases}.
\end{eqnarray}
\hrulefill
\end{figure*}
Note that the same or different prices can be applied to the miners, which can be referred as the uniform and discriminatory pricing schemes, respectively. In the following, we investigate these two pricing schemes for resource management in mobile blockchain.

\section{Equilibrium Analysis for Edge Computing Resource Management}\label{Sec:Solution}
In this section, the uniform pricing and discriminatory pricing schemes are proposed for resource management in mobile blockchain. Then, under both pricing schemes, we analyze the optimal service demand of miners as well as the profit maximization of the ESP.

\subsection{Uniform Pricing Scheme}\label{SubSec:Uniform}
We first consider the uniform pricing scheme, in which the ESP charges all the miners the same price per unit for their edge computing service demand, i.e., $p_i = p, \forall i$. Given the payoff functions defined in Section~\ref{Sec:Model}, we apply backward induction to investigate the Stackelberg game.
\subsubsection{Stage II: Miners' Demand Game}\label{SubsubSec:UtilityMaximization_Uniform}
Given the price $p$ decided by the ESP, in Stage II, the miners compete with each other to maximize their own utility by choosing their individual service demand, which forms the noncooperative Miners' Demand Game (MDG)~$\mathcal{G}^u= \{\mathcal{N},\{x_i\}_{i \in \mathcal{N}},\{u_i\}_{i \in \mathcal{N}}\}$, where $\cal N$ is the set of miners, $\{x_i\}_{i \in \mathcal{N}}$ is the strategy set, and $u_i$ is the payoff function of miner $i$. 


\begin{theorem}
The Nash equilibrium in MDG~$\mathcal{G}^u= \{\mathcal{N},\{x_i\}_{i \in \mathcal{N}},\{u_i\}_{i \in \mathcal{N}}\}$ exists.
\end{theorem}
\begin{proof}
Firstly, the strategy space for each miner, $[\underline x ,\overline x]$ is a non-empty, convex, compact subset of the Euclidean space. Further, we know $u_i$ is apparently continuous in $[\underline x ,\overline x]$. Then, we derive the first order and second order derivatives of~(\ref{Eq:Utility}) with respect to $x_i$, which can be written as follows:
\begin{equation}\label{Eq:FirstOrderUtility}
\frac{{\partial {u_i}}}{{\partial {x_i}}} = (R + r{t_i}){e^{ - \lambda z{t_i}}}\frac{{\partial {\alpha_i}}}{{\partial {x_i}}} - p,
\end{equation}
\begin{equation}\label{Eq:SecondOrderUtility}
\frac{{{\partial ^2}{u_i}}}{{\partial {x_i}^2}} = (R + r{t_i}){e^{ - \lambda z{t_i}}}\frac{{{\partial ^2}{\alpha_i}}}{{\partial {x_i}^2}} < 0,
\end{equation}
where $\frac{{\partial {\alpha_i}}}{{\partial {x_i}}} = \frac{{\sum\nolimits_{i \ne j} {{x_j}} }}{{{{\left( {\sum\nolimits_{i\in \cal N} {{x_j}} } \right)}^2}}} > 0, \frac{{{\partial ^2}{\alpha_i}}}{{\partial {x_i}^2}} = - 2\frac{{\sum\nolimits_{i \ne j} {{x_j}} }}{{{{\left( {\sum\nolimits_{i \in \cal N} {{x_j}} } \right)}^3}}} < 0$. Therefore, we have proved that $u_i$ is strictly concave with respect to $x_i$. Accordingly, the Nash equilibrium of noncooperative MDG $\mathcal{G}^u$ exists~\cite{han2012game}. The proof is completed.
\end{proof}

\begin{theorem}
The uniqueness of the Nash equilibrium in the noncooperative MDG is guaranteed provided that the following condition
\begin{equation}\label{Eq:Condition}
\begin{footnotesize}
\frac{{2(N - 1){e^{ - \lambda z{t_i}}}}}{{R + r{t_i}}} < \sum\nolimits_{i \in \cal N} {\frac{{{e^{ - \lambda z{t_i}}}}}{{R + r{t_i}}}}
\end{footnotesize}
\end{equation}
is ensured.
\end{theorem}
\begin{proof}
Please refer to the appendix for details.
\end{proof}
\begin{theorem}
The unique Nash equilibrium for miner $i$ in the MDG is given by
\begin{equation}\label{Eq:MDGequilibruim}
{x_i}^* = \frac{{N - 1}}{{\sum\nolimits_{j\in \cal N} {\frac{{p{e^{ - \lambda z{t_j}}}}}{{R + r{t_j}}}} }} - {\left( {\frac{{N - 1}}{{\sum\nolimits_{j\in \cal N} {\frac{{p{e^{ - \lambda z{t_j}}}}}{{R + r{t_j}}}} }}} \right)^2}\frac{{p{e^{ - \lambda z{t_i}}}}}{{R + r{t_i}}}, \forall i,
\end{equation}
given the condition in~(\ref{Eq:Condition}) holds.
\end{theorem}
\begin{proof}
Please refer to the appendix for details.
\end{proof}
Therefore, we can apply the best-response dynamics to obtain the Nash equilibrium of the N-player noncooperative game in Stage II~\cite{han2012game}.

\subsubsection{Stage I: ESP's Profit Maximization}\label{SubsubSec:ProfitMaximization_Uniform}
Based on the Nash equilibrium of the computing service demand in the MDG~$\mathcal{G}^u$ in Stage II, the ESP, i.e., the leader can optimize its pricing strategy in Stage I to maximize its profit defined in~(\ref{Eq:Profit}). 
We substitute~(\ref{Eq:MDGequilibruim}) into~(\ref{Eq:Profit}), and the profit maximization of the ESP is simplified as follows:
\begin{equation}\label{Eq:Profit2}
\begin{footnotesize}
\begin{aligned}
& \underset{p>0}{\text{maximize}}
& & \Pi(p) = (p - cT)\frac{{N - 1}}{{\sum\nolimits_{j \in \cal N} {\frac{{p{e^{ - \lambda z{t_j}}}}}{{R + r{t_j}}}} }} \\
& \text{subject to}
& & 0\le p \le \overline p.\\
\end{aligned}
\end{footnotesize}
\end{equation}
\begin{theorem}
Under uniform pricing, the ESP achieves profit maximization, under the unique optimal price.
\end{theorem}
\begin{proof}
From~(\ref{Eq:Profit2}), we have $\Pi(p) = \frac{{p - cT}}{p}\frac{{N - 1}}{{\sum\nolimits_{j \in \cal N} {\frac{{{e^{ - \lambda z{t_j}}}}}{{R + r{t_j}}}} }}$. The first and second derivatives of profit $\Pi(p)$ with respect to price $p$ are given as follows:
\begin{equation}\label{Eq:FirstOrderProfit}
\begin{footnotesize}
\frac{{d \Pi(p) }}{{d p}} = \frac{{cT}}{{{p^2}}}\frac{{N - 1}}{{\sum\nolimits_{j \in \cal N} {\frac{{{e^{ - \lambda z{t_j}}}}}{{R + r{t_j}}}} }},
\end{footnotesize}
\end{equation}
\begin{equation}\label{Eq:SecondOrderProfit}
\begin{footnotesize}
\frac{{{d ^2}\Pi(p) }}{{d {p^2}}} = - \frac{{2cT}}{{{p^2}}}\frac{{N - 1}}{{\sum\nolimits_{j \in \cal N} {\frac{{{e^{ - \lambda z{t_j}}}}}{{R + r{t_j}}}} }}< 0.
\end{footnotesize}
\end{equation}
Due to the negativity of~(\ref{Eq:SecondOrderProfit}), the strict concavity of the objective function is ensured. Thus, the ESP is able to achieve the maximum profit with the unique optimal price. The proof is completed.
\end{proof}
Under uniform pricing, we have proved that the Nash equilibrium in Stage II is unique and the optimal price in Stage I is also unique. Thus, we can conclude that the Stackelberg equilibrium is unique and accordingly the best-response dynamics algorithm can achieve this unique Stackelberg equilibrium~\cite{han2012game}. 
\subsection{Discriminatory Pricing Scheme}\label{SubSec:DiscriminatoryPricing}
We next study the discriminatory pricing scheme, where the ESP is able to set different prices per unit for service demand from different miners. 

\subsubsection{Stage II: Miners' Demand Game}\label{SubsubSec:UtilityMaximization_Discriminatory}
Under the discriminatory pricing scheme, the strategy space of the ESP becomes $\{{\bf p}= [p_i]_{i\in \cal N}:0\le p_i \le \overline p\}$. Recall that we prove the existence and uniqueness of MDG~$\mathcal{G}^u$, given the fixed price from the ESP. Thus, under discriminatory pricing, the existence and uniqueness of the MDG can be still guaranteed. With minor change from Theorem~3, we have the following theorem accordingly.
\begin{theorem}
Under uniform pricing, the unique Nash equilibrium demand of miner $i$ can be obtained as follows:
\begin{equation}\label{Eq:MDGequilibruim2}
\begin{footnotesize}
{x_i}^* = \frac{{N - 1}}{{\sum\nolimits_{j\in \cal N}{\frac{{p_j{e^{ - \lambda z{t_j}}}}}{{R + r{t_j}}}} }} - {\left( {\frac{{N - 1}}{{\sum\nolimits_{j \in \cal N} {\frac{{p_j{e^{ - \lambda z{t_j}}}}}{{R + r{t_j}}}} }}} \right)^2}\frac{{p_i{e^{ - \lambda z{t_i}}}}}{{R + r{t_i}}}, \forall i,
\end{footnotesize}
\end{equation}
provided that the following condition
\begin{equation}\label{Eq:Condition3}
\begin{footnotesize}
\frac{{2(N - 1){p_i}{e^{ - \lambda z{t_i}}}}}{{R + r{t_i}}} < \sum\nolimits_{j\in \cal N} {\frac{{{p_j}{e^{ - \lambda z{t_j}}}}}{{R + r{t_j}}}}
\end{footnotesize}
\end{equation}
holds.
\end{theorem}
\begin{proof}
The steps of proof are similar to those in the case of uniform pricing as shown in Section~\ref{SubsubSec:UtilityMaximization_Uniform}, and thus we omit them for brevity.
\end{proof}


\subsubsection{Stage I: ESP's Profit Maximization}\label{SubsubSec:ProfitMaximization_Discriminatory}

Similar to that in Section~\ref{SubsubSec:ProfitMaximization_Uniform}, we analyze the profit maximization with the Nash equilibrium of the computing service demand in Stage II. After substituting~(\ref{Eq:MDGequilibruim2}) into~(\ref{Eq:Profit}), we have the following optimization,
\begin{equation}\label{Eq:Profit3}
\begin{footnotesize}
\begin{aligned}
& \underset{\bf p>0}{\text{maximize}}
& & \Pi({\bf p}) = \sum\nolimits_{i \in \cal N}\left(p_i - cT\frac{{N - 1}}{{\sum\nolimits_{j\in \cal N} {\frac{{p_j{e^{ - \lambda z{t_j}}}}}{{R + r{t_j}}}} }}\right) \\
& \text{subject to}
& & 0\le p_i \le \overline p, \forall i.\\
\end{aligned}
\end{footnotesize}
\end{equation}
\begin{figure*}[ht]\scriptsize
\begin{equation}\label{Eq:PriceFinal}
g({\bf p}) = \sum\nolimits_{j \ne h} {\left( {{a_h}\left( {1 - \frac{{{p_h}}}{{{a_h}}}\frac{{N - 1}}{{\sum\nolimits_{h \in \cal N} {\frac{{{p_h}}}{{{a_h}}}} }}} \right)\left( {1 - \frac{{{p_j}}}{{{a_j}}}\frac{{N - 1}}{{\sum\nolimits_{h \in \cal N} {\frac{{{p_h}}}{{{a_h}}}} }}} \right)} \right)}.
\end{equation}
\begin{eqnarray}\label{Eq:PriceFinalFirstOrder}
\frac{{\partial g({\bf{p}})}}{{\partial {p_i}}} = \sum\nolimits_{j \ne i} {\left( {\left( {{a_i} + {a_j}} \right)\left( {\frac{{ - \frac{{N - 1}}{{{a_i}}}\sum\nolimits_{h \ne i} {\frac{{{p_h}}}{{{a_h}}}} }}{{{{\left( {\sum\nolimits_{h \in \cal N} {\frac{{{p_h}}}{{{a_h}}}} } \right)}^2}}}\left( {1 - \frac{{N - 1}}{{\sum\nolimits_{h \in \cal N} {\frac{{{p_h}}}{{{a_h}}}} }}\frac{{{p_j}}}{{{a_j}}}} \right) + \frac{{\frac{{N - 1}}{{{a_i}}}\frac{{{p_j}}}{{{a_j}}}}}{{{{\left( {\sum\nolimits_{h \in \cal N} {\frac{{{p_h}}}{{{a_h}}}} } \right)}^2}}}\left( {1 - \frac{{N - 1}}{{\sum\nolimits_{h \in \cal N} {\frac{{{p_h}}}{{{a_h}}}} }}\frac{{{p_i}}}{{{a_i}}}} \right)} \right)} \right)}.
\end{eqnarray}
\begin{equation}\label{Eq:PriceFinalSecondOrder}
\frac{{{\partial ^2}g({\bf{p}})}}{{\partial {p_i}^2}} 
= \sum\nolimits_{j \ne i} {\left( {\left( {{a_i} + {a_j}} \right)\left( {\frac{{2\frac{{N - 1}}{{{a_i}^2}}\sum\nolimits_{h \ne i} {\frac{{{p_h}}}{{{a_h}}}} }}{{{{\left( {\sum\nolimits_{h \in \cal N} {\frac{{{p_h}}}{{{a_h}}}} } \right)}^3}}}\left( {1 - 2\frac{{N - 1}}{{\sum\nolimits_{h \in \cal N} {\frac{{{p_h}}}{{{a_h}}}} }}\frac{{{p_j}}}{{{a_j}}}} \right) - \frac{{2\frac{{N - 1}}{{{a_i}^2}}\frac{{{p_j}}}{{{a_j}}}}}{{{{\left( {\sum\nolimits_{h \in \cal N} {\frac{{{p_h}}}{{{a_h}}}} } \right)}^3}}}\left( {1 - \frac{{N - 1}}{{\sum\nolimits_{h \in \cal N} {\frac{{{p_h}}}{{{a_h}}}} }}\frac{{{p_i}}}{{{a_i}}}} \right)} \right)} \right)}.
\end{equation}
\hrulefill
\end{figure*}

\begin{theorem}
$\Pi({\bf p})$ is concave on each $p_i$, when $\sum\nolimits_{i \ne j} {({a_i} + {a_j})\left( {1 - \frac{{N\frac{{{p_j}}}{{{a_j}}}}}{{\sum\nolimits_{j \in \cal N} {\frac{{{p_j}}}{{{a_j}}}} }}} \right)} \le 0$, and decreasing on each $p_i$ when $\sum\nolimits_{i \ne j} {({a_i} + {a_j})\left( {1 - \frac{{N\frac{{{p_j}}}{{{a_j}}}}}{{\sum\nolimits_{j \in \cal N} {\frac{{{p_j}}}{{{a_j}}}} }}} \right)} >0$, provided that the following condition
\begin{equation}\label{Eq:Condition4}
\begin{footnotesize}
\frac{{{p_i}}}{{{a_i}}} \ge \frac{{\sum\nolimits_{j \in \cal N} {\frac{{{p_j}}}{{{a_j}}}} }}{{{{(N - 1)}^2}}}
\end{footnotesize}
\end{equation}
is ensured, where $a_i =(R+rt_i)e^{-\lambda z t_i}$.
\end{theorem}
\begin{proof}[Sketch of Proof]
We firstly decompose the objective function in~(\ref{Eq:Profit3}) into two components, namely, $\sum\nolimits_{i} c T{x_i^*}$ and ${\sum\nolimits_{i} p_i}{x_i^*}$. Then, we analyze the properties of each component. We define
\begin{equation}\label{Eq:cost}
\begin{footnotesize}
f({\bf p}) = - cT{x_i^*} =  - cT\frac{{N - 1}}{{\sum\nolimits_{j \in \cal N} {\frac{{p_j{e^{ - \lambda z{t_j}}}}}{{R + r{t_j}}}} }}.
\end{footnotesize}
\end{equation}
Let $a_j =(R+rt_j)e^{-\lambda z t_j}$, and we have $f({\bf p}) = \frac{{ - cT(N - 1)}}{{\sum\nolimits_{j \in \cal N} {\frac{{{p_j}}}{{{a_j}}}} }}$. We next obtain the first and the second partial derivatives of~(\ref{Eq:cost}) with respect to $p_i$ as follows:
\begin{equation}
\begin{footnotesize}
\frac{{\partial f({\bf{p}})}}{{\partial {p_i}}} = \frac{{(N - 1)cT}}{{{a_i}{{\left( {\sum\nolimits_{j \in \cal N} {\frac{{{p_j}}}{{{a_j}}}} } \right)}^2}}}, \frac{{{\partial ^2}f({\bf{p}})}}{{\partial {p_i}^2}} = \frac{{ - 2(N - 1)cT}}{{{a_i}^2{{\left( {\sum\nolimits_{j \in \cal N} {\frac{{{p_j}}}{{{a_j}}}} } \right)}^3}}}.
\end{footnotesize}
\end{equation}
Further, we have $\frac{{\partial f({\bf{p}})}}{{\partial {p_i}{p_j}}} = \frac{{ - 2(N - 1)cT}}{{{a_i}{a_j}{{\left( {\sum\nolimits_{j\in \cal N} {\frac{{{p_j}}}{{{a_j}}}} } \right)}^3}}}$. Accordingly, we can prove that the Hessian matrix of $f({\bf p})$ is semi-negative definite. The detailed proof is given in~\cite{xiong2017blockchain} for the space limit.

Then, we analyze the properties of ${\sum\nolimits_{i} p_i}{x_i^*}$. We first define
\begin{equation}\label{Eq:Price}
\begin{footnotesize}
g({\bf p}) = \sum\nolimits_{i \in \cal N} {{p_i}} {x_i}^* = \frac{{\sum\nolimits_{j \ne i} {{a_i}{x_i}{x_j}} }}{{{{\left( {\sum\nolimits_{j \ne i} {{x_j}} } \right)}^2}}}.
\end{footnotesize}
\end{equation}
By substituting~(\ref{Eq:MDGequilibruim2}) into~(\ref{Eq:Price}), we can obtain the final expression for $g({\bf p})$, which can be rewritten as in~(\ref{Eq:PriceFinal}). Then, we derive the first order and the second partial derivatives of~(\ref{Eq:PriceFinal}) with respect to $p_i$ as shown in~(\ref{Eq:PriceFinalFirstOrder}) and~(\ref{Eq:PriceFinalSecondOrder}). Since we have $x_i = \frac{{N - 1}}{{\sum\nolimits_{h \in \cal N} {\frac{{{p_h}}}{{{a_h}}}} }} - \frac{{{p_i}}}{{{a_i}}}{\left( {\frac{{N - 1}}{{\sum\nolimits_{h \in \cal N} {\frac{{{p_h}}}{{{a_h}}}} }}} \right)^2} = \frac{{N - 1}}{{\sum\nolimits_{h \in \cal N} {\frac{{{p_h}}}{{{a_h}}}} }}\left( {1 - \frac{{N - 1}}{{\sum\nolimits_{h \in \cal N} {\frac{{{p_h}}}{{{a_h}}}} }}\frac{{{p_i}}}{{{a_i}}}} \right) >0 $ and ${1 - \frac{{N - 1}}{{\sum\nolimits_{h \in \cal N} {\frac{{{p_h}}}{{{a_h}}}} }}\frac{{{p_i}}}{{{a_i}}}} >0$. Next we define ${\delta _{ij}} = \sum\nolimits_{i \ne j} {({a_i} + {a_j})\left( {1 - \frac{{N\frac{{{p_j}}}{{{a_j}}}}}{{\sum\nolimits_{j\in \cal N} {\frac{{{p_j}}}{{{a_j}}}} }}} \right)}$ for the ease of presentation. When ${\delta _{ij}} \le 0$, we can easily prove that $\frac{{{\partial ^2}g({\bf{p}})}}{{\partial {p_i}^2}} < 0$, i.e., $g({\bf p})$ is concave on each $p_i$.  Next, we mainly prove that $\Pi({\bf p})$ is a monotonically decreasing function with respect to $p_i$, when ${\delta _{ij}} > 0$. The detailed proof is given in~\cite{xiong2017blockchain} for the space limit.
\end{proof}
\begin{theorem}
Under discriminatory pricing, the ESP achieves the profit maximization by finding the unique optimal prices.
\end{theorem}
\begin{proof}[Sketch of Proof]
From Theorem~6, we conclude that
when $\Pi({\bf p})$ is concave on $p_i$, $p_i$ needs to be smaller than a certain threshold, and $\Pi({\bf p})$ is decreasing on $p_i$ when $p_i$ is larger than this threshold. Therefore, it can be concluded that the optimal value of profit of the ESP, i.e., $\Pi^*({\bf p})$ is achieved in the concave parts when ${\delta _{ij}} \le 0$. Apparaently, the maximization of profit $\Pi({\bf p})$ is achieved either in the boundary of domain area or in the local maximization point. Since we know that the optimal value of profit, i.e., $\Pi^*({\bf p})$ is achieved in the interior area, and thus $\bf p^*$ exists. In the following, we prove that there exists at most one optimal solution by using Variational Inequality theory~\cite{scutari2010convex}.

Let the set ${\cal K}=\big \{ {\bf p} = [p_1, \ldots, p_N]^{\top}\big| {\delta _{ij}} \le 0, \forall i \in \cal N \big\}$. 
Then, we formulate an equivalent problem to~(\ref{Eq:Profit3}) as follows:
\begin{equation}\label{Eq:optimizationtoVI}
\begin{footnotesize}
\begin{aligned}
& \underset{\bf p>0}{\text{minimize}}
& & -\Pi({\bf p}) \\
& \text{subject to}
& & {\bf p} \in {\cal K}.\\
\end{aligned}
\end{footnotesize}
\end{equation}
Let $F({\bf p})=\nabla \left( { - \Pi ({\bf{p}})} \right) = - {\left[ {{\nabla _{{p_i}}}\Pi } \right]^{\top}_{i \in \cal N}}$. Accordingly, the optimization problem in~(\ref{Eq:optimizationtoVI}) corresponds to finding a point set ${\bf p^*}\in \cal K$, such that $({\bf p} - {\bf p^*})F({\bf p^*})\ge 0 ,\forall {\bf p} \in \cal K$, which is the Variational Inequality (VI) problem: VI$({\cal K}, F)$. Next, we mainly prove that $F$ is strictly monotone on $\cal K$ and continuous, and thus it can be concluded that VI$({\cal K}, F)$ has at most one solution according to~\cite{scutari2010convex}. Thus, the equivalent problem admits at most one solution, from which the uniqueness of the optimal solution, i.e., the Stackelberg equilibrium, follows. The detailed proof is given in~\cite{xiong2017blockchain} for the space limit.
\end{proof}

Therefore, we can apply the low-complexity gradient based algorithm to achieve the maximized profit $\Pi({\bf p})$ of the ESP. Similar to that in Section~\ref{SubSec:Uniform}, we adopt the best response dynamics algorithm to obtain the unique Stackelberg equilibrium, under which the ESP achieves the profit maximization according to Theorem~7. 

\section{Performance Evaluation}\label{Sec:Simulation}

\begin{figure}[t]
\centering
\includegraphics[width=.352\textwidth]{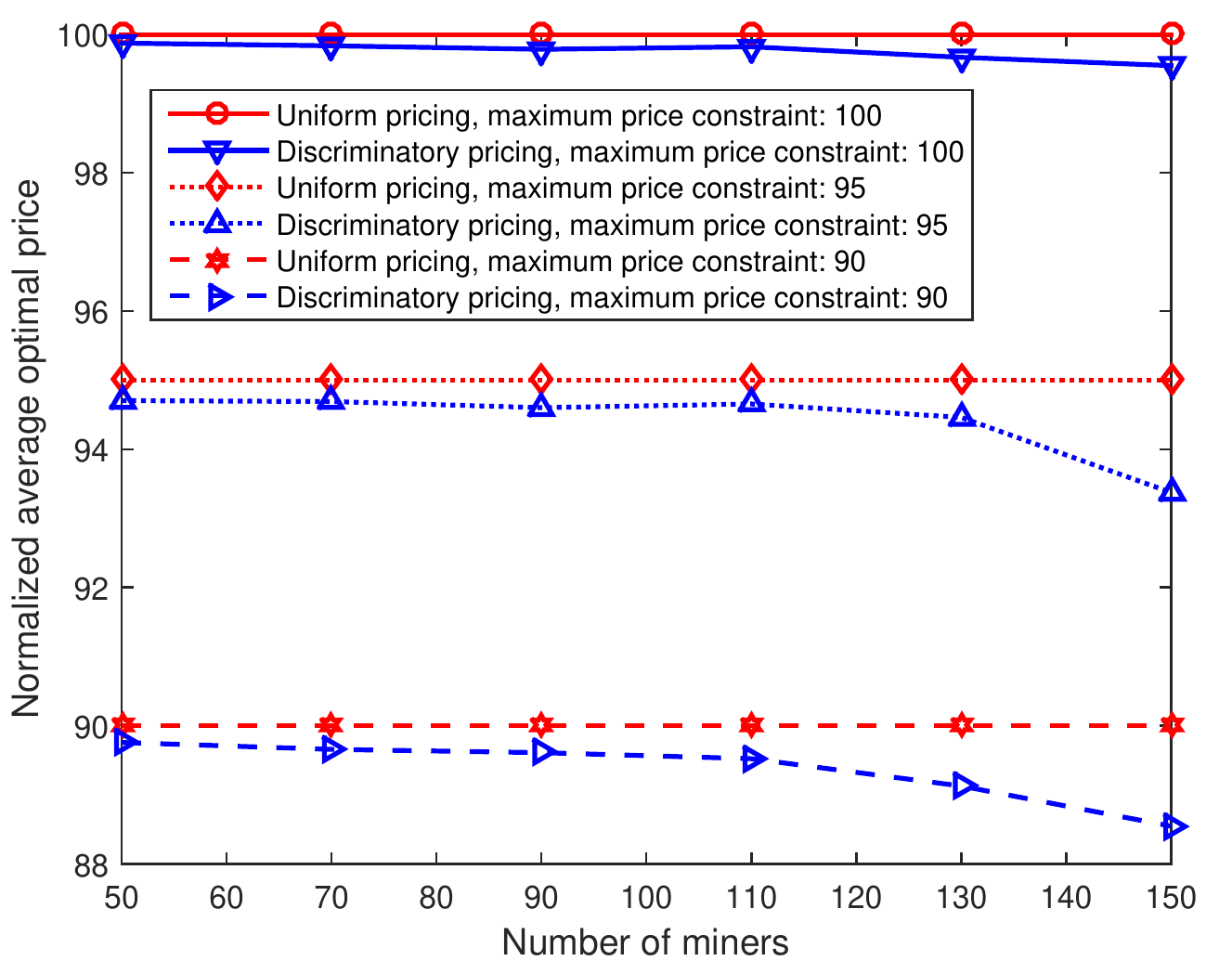}
\caption{Normalized average optimal price versus the number of miners.}\label{Fig:price_comparison}
\end{figure}
In this section, we conduct the numerical simulations to evaluate the performance of our proposed optimal pricing-based resource management in mobile blockchain. We consider a group of $N$ mobile users, i.e., miners in mobile blockchain mining for the reward and further assume that the size of a block to be mined by miner $i$ follows the normal distribution ${\cal N}(\mu_t, \sigma^2)$. The default parameters are set as follows: $\underline x = 10^{-2}$, $\overline x = 100$, $N=100$, $\overline p =100$, $\mu_t = 200$, $\sigma^2 = 5$, $R = 10^4$, $r = 20$, $z = 5 \times 10^{-3}$ and $c=10^{-3}$. 
\begin{figure*}
\begin{minipage}[t]{0.50\textwidth}
\centering
\includegraphics[width=2.2in]{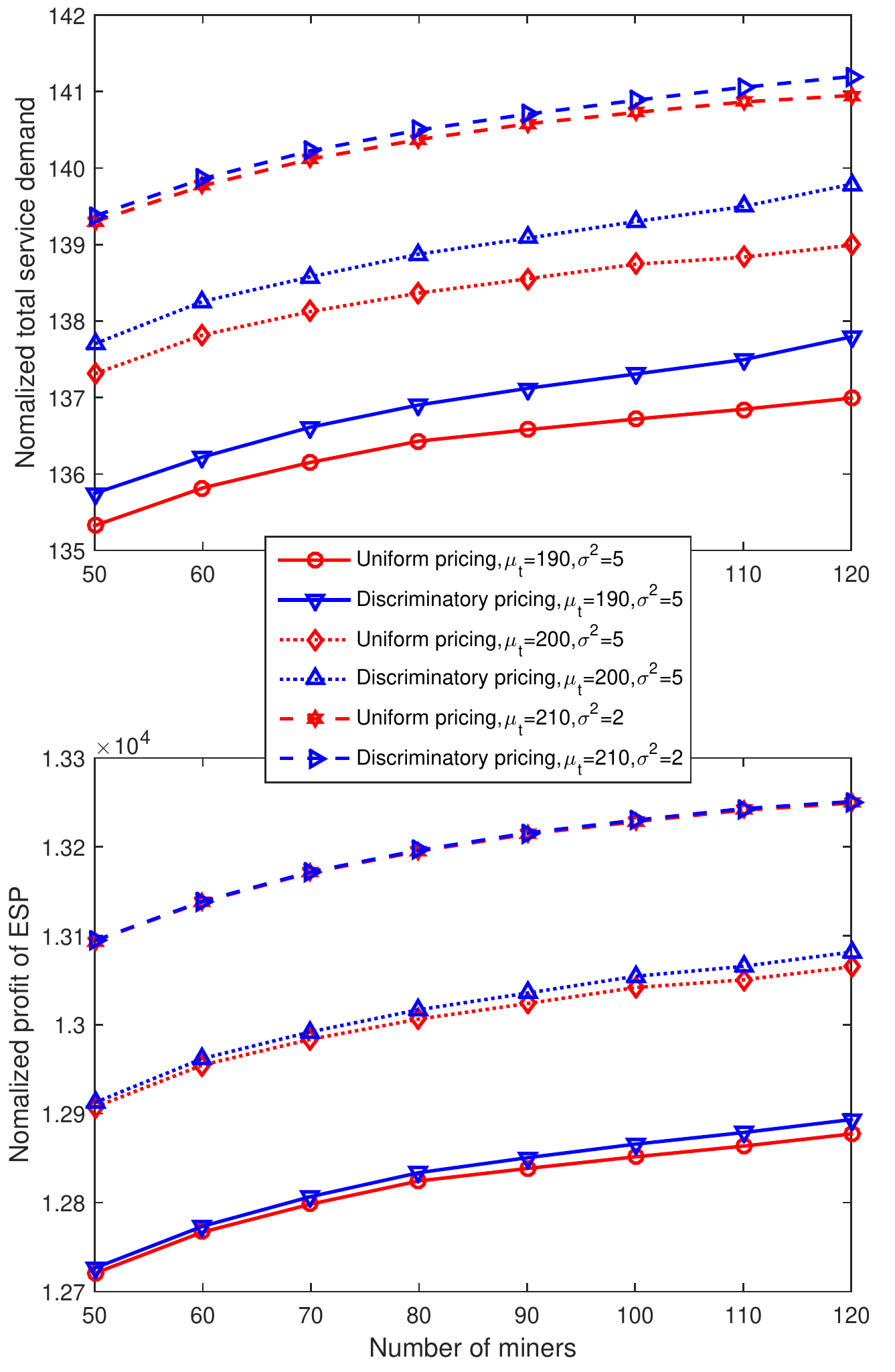}
\caption{Normalized total service demand of miners and the profit of the ESP versus the number of miners.}\label{Fig:utility_profit_number}
\end{minipage}
\begin{minipage}[t]{0.50\textwidth}
\centering
\includegraphics[width=2.2in]{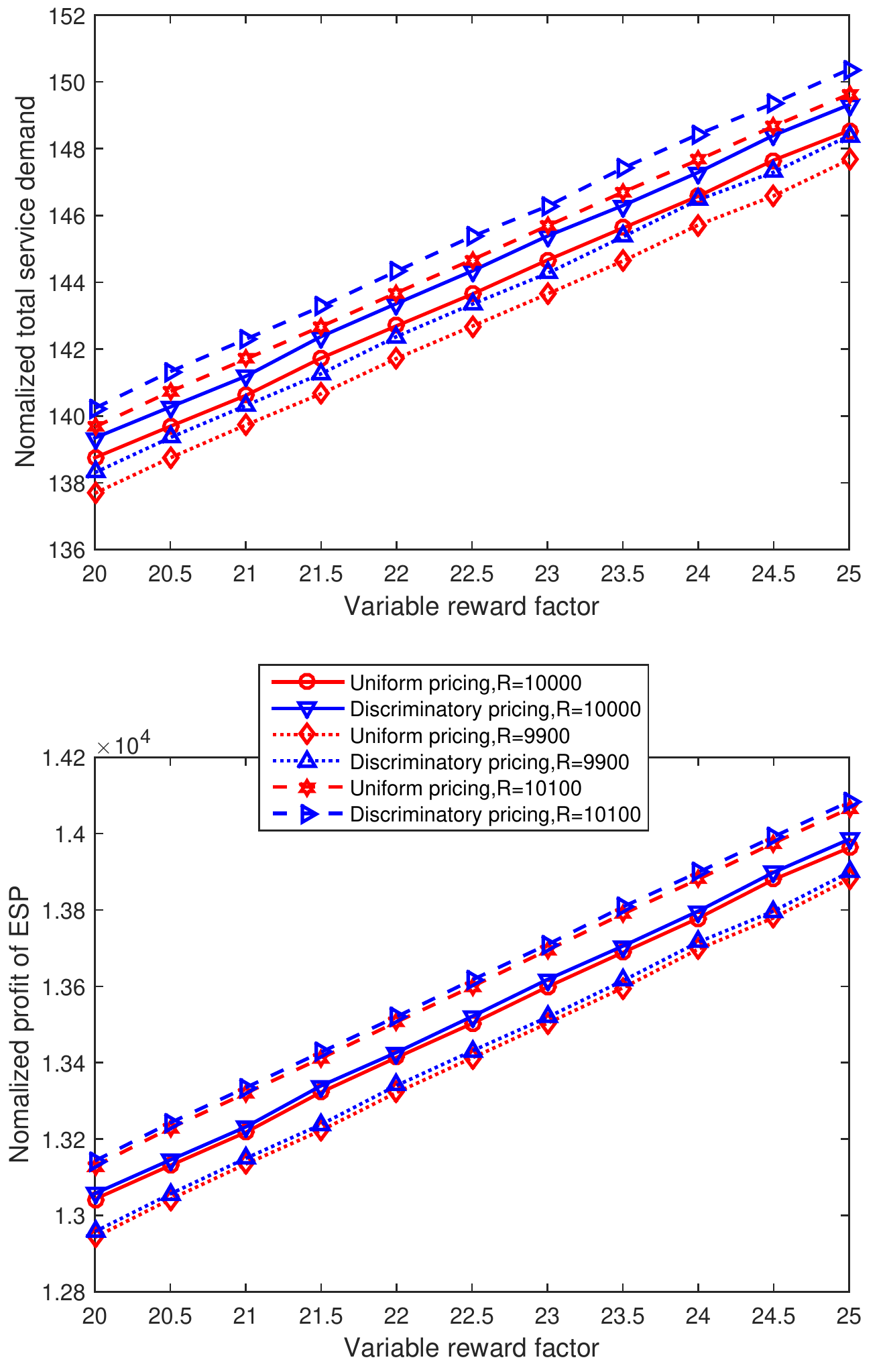}
\caption{Normalized total service demand of miners and the profit of the ESP versus the variable reward factor.}\label{Fig:utility_profit_reward}
\end{minipage}
\end{figure*}

We first evaluate the comparison of proposed uniform pricing and discriminatory pricing schemes. Figure.~\ref{Fig:price_comparison} depicts the comparison of the normalized average optimal price under two proposed pricing schemes. It is found that the optimal price under uniform pricing is the same as the maximum price, which is in line with~(\ref{Eq:FirstOrderProfit}). Specifically, the expression in~(\ref{Eq:FirstOrderProfit}) is always positive, and thus the profit of the ESP increases with the increase of price. In other words, the maximum price is the optimal value for profit maximization of the ESP under uniform pricing. 
Further, we find that the average optimal price of discriminatory pricing is slightly lower than that of uniform pricing. The intuitive reason is that, the ESP can set different unit prices of service for different miners under discriminatory pricing scheme. Therefore, the ESP can significantly encourage the higher total service demand from miners and achieve greater profit, which is also in line with the following results. As shown in Figs.~\ref{Fig:utility_profit_number} and~\ref{Fig:utility_profit_reward}, the total service demand from miners and the profit of the ESP under the uniform pricing scheme is slightly lower than that under the discriminatory pricing scheme in all cases. 
From Fig.~\ref{Fig:utility_profit_number}, we find that as $\sigma^2$ decreases, the results under uniform pricing scheme approach to those under discriminatory pricing more. This is because the heterogeneity of miners in blockchain is reduced as $\sigma^2$ decreases. We may consider one symmetric case, where the miners are homogeneous with the same size of blocks to mine, i.e., $\sigma^2 = 0$. In this case, the discriminatory pricing scheme yields the same results as those of the uniform pricing scheme.

We next evaluate the impacts of the number of miners. From Fig.~\ref{Fig:utility_profit_number}, we find that the total service demand of miners and the profit of the ESP increase with the increase of the total number of miners in mobile blockchain. This is because adding more miners will intensify the competition among the miners, which potentially motivates them to have higher service demand. Further, the coming miners have their service demand, and thus the total service demand from miners is increased. In turn, the ESP extracts more surplus from miners and thus has greater profit. In addition, it is found that the rate of service demand increment decreases as the total number of miners increases. This is because the incentive of miners to increase their service demand is weakened because the probability of their successful mining is reduced when the total number of miners is increasing. It is also found that with the increase of $\mu_t$, the total service demand of miners and the profit of the ESP increase. This is because as $\mu_t$ increases, i.e., the average size of one block becomes larger, the variable reward for each miner also increases. The potential incentive of miners to increase their service demand is improved, and thus the total service demand of miners increases. Accordingly, the ESP achieves greater profit.

At last, we examine the impacts brought by the variable reward and fixed reward, which are shown in Fig.~\ref{Fig:utility_profit_reward}. We find that both the total service demand of miners and the profit of the ESP increase with the increase of variable reward factor. This is because the increased variable reward enhances the motivation of miners for higher service demand, and the total service demand is improved accordingly. Correspondingly, the ESP achieves greater profit. Further, by comparing curves with different values of fixed reward, it is also found that as the fixed reward increases, the total service demand of miners and the profit of the ESP also increase. Similarly, this is because the increased fixed reward generates greater incentive of miners, which in turn improves the total service demand of miners and the profit of the ESP.

\begin{figure*}[ht]\scriptsize
\begin{eqnarray}\label{Eq:Monotonicity}
{\cal F}_i({\bf{x}}') - {\cal F}_i({\bf{x}}) &=& \sqrt {\frac{{R + r{t_i}\sum\nolimits_{i \ne j} {{x_j^\prime} } }}{{p{e^{ - \lambda z{t_i}}}}}} - {\sum\nolimits_{i \ne j} {{x_j^\prime}} } - \sqrt {\frac{{(R + r{t_i})\sum\nolimits_{i \ne j} {{x_j}} }}{{p{e^{ - \lambda z{t_i}}}}}} - \sum\nolimits_{i \ne j} {{x_j}} \nonumber \\ &=& \sqrt {\frac{{R + r{t_i}}}{{p{e^{ - \lambda z{t_i}}}}}} \left( {\sqrt {{{\sum\nolimits_{i \ne j} {{x_j^\prime}} } }} - \sqrt {\sum\nolimits_{i \ne j} {{x_j}} } } \right)- \left( {{{\sum\nolimits_{i \ne j} {{x_j^\prime}} } } - \sum\nolimits_{i \ne j} {{x_j}} } \right)\nonumber \\ &=& \left( {\sqrt {\frac{{R + r{t_i}}}{{p{e^{ - \lambda z{t_i}}}}}} - \sqrt {{{\sum\nolimits_{i \ne j} {{x_j^\prime}} } }} - \sqrt {\sum\nolimits_{i \ne j} {{x_j}} } } \right) \left( {\sqrt {{{\sum\nolimits_{i \ne j} {{x_j^\prime}} } }} - \sqrt {\sum\nolimits_{i \ne j} {{x_j}} } } \right).
\end{eqnarray}
\begin{eqnarray}\label{Eq:Scalability}
\lambda {\cal F}_i({\bf x}) - {\cal F}_i(\lambda {\bf x}) &=& \lambda \sqrt {\frac{{(R + r{t_i})\sum\nolimits_{i \ne j} {{x_j}} }}{{p{e^{ - \lambda z{t_i}}}}}} - \lambda \sum\nolimits_{i \ne j} {{x_j}} - \sqrt {\frac{{(R + r{t_i})\sum\nolimits_{i \ne j} {\lambda {x_j}} }}{{p{e^{ - \lambda z{t_i}}}}}} - \sum\nolimits_{i \ne j} {\lambda {x_j}}\nonumber \\
 &=& \left( {\lambda - \sqrt \lambda } \right)\sqrt {\frac{{(R + r{t_i})\sum\nolimits_{i \ne j} {{x_j}} }}{{p{e^{ - \lambda z{t_i}}}}}} > 0, \forall \lambda>1.
\end{eqnarray}
\hrulefill
\end{figure*}

\section{Conclusion}\label{Sec:Conclusion}

In this paper, we have addressed the optimal pricing-based edge computing resource management, in order to support offloading for mobile blockchain mining. In particular, we have formulated the Stackelberg game model to jointly study the profit maximization of edge computing service provider and the utility maximization of miners. In the model, we have derived the unique Nash equilibrium point of the game among the miners through backward induction. Further, we have analyzed the resource management including the uniform and discriminatory pricing schemes for the edge computing service provider. We have proved the existence and uniqueness of the Stackelberg equilibrium analytically for both pricing schemes. Additionally, we have conducted the simulations to evaluate the network performance, which may guide the edge computing service provider to achieve optimal resource management.
\appendix
\textbf{Proof of Theorem 2:}
\begin{proof}
Based on the first order derivative condition in~(\ref{Eq:FirstOrderUtility}), we obtain the best response function of miner $i$, as shown in~(\ref{Eq:BestResponse}). The uniqueness of the Nash equilibrium can be proved provided that the best response function of miner $i$, i.e., as given in~(\ref{Eq:BestResponse}), is the standard function~\cite{han2012game}.

Firstly, for the positivity, under the condition in~(\ref{Eq:Condition}), we have (from Lemma~1 given in~\cite{xiong2017blockchain}) $\sum\nolimits_{i \ne j} {{x_j}} < \frac{{R + r{t_i}}}{{p{e^{ - \lambda z{t_i}}}}} < \frac{{R + r{t_i}}}{{4p{e^{ - \lambda z{t_i}}}}}$, then we can conclude that $\sum\nolimits_{i \ne j} {{x_j}} < \sqrt {\frac{{R + r{t_i}\sum\nolimits_{i \ne j} {{x_j}} }}{{p{e^{ - \lambda z{t_i}}}}}}$. Thus, we can prove that ${\mathscr F}_i({\bf x}) = \sqrt {\frac{{R + r{t_i}\sum\nolimits_{i \ne j} {{x_j}} }}{{p{e^{ - \lambda z{t_i}}}}}} - \sum\nolimits_{i \ne j} {{x_j}} > 0$, which is the positivity condition.

Secondly, we show the monotonicity of~(\ref{Eq:BestResponse}). Let $\bf x' > x$, we can further simplify the expression of ${\cal F}_i({\bf{x}}') - {\cal F}_i({\bf{x}})$, which is shown in~(\ref{Eq:Monotonicity}). In particular, we have ${\sqrt {{{\sum\nolimits_{i \ne j} {{x_j^\prime}} } }} - \sqrt {\sum\nolimits_{i \ne j} {{x_j}} } }>0$, and we can easily verify that $ \sqrt {\frac{{R + r{t_i}}}{{p{e^{ - \lambda z{t_i}}}}}} - \sqrt {{{\sum\nolimits_{i \ne j} {{x_j^\prime}} } }} - \sqrt {\sum\nolimits_{i \ne j} {{x_j}} } \in\left( {\sqrt {\frac{{R + r{t_i}}}{{p{e^{ - \lambda z{t_i}}}}}} - 2\sqrt {{{\sum\nolimits_{i \ne j} {{x_j^\prime}} } }} ,\sqrt {\frac{{R + r{t_i}}}{{p{e^{ - \lambda z{t_i}}}}}} - 2\sqrt {\sum\nolimits_{i \ne j} {{x_j}} } } \right)$. Given Lemma~1 given in~\cite{xiong2017blockchain}, we can prove that ${\sqrt {\frac{{R + r{t_i}}}{{p{e^{ - \lambda z{t_i}}}}}} - 2\sqrt {\sum\nolimits_{i \ne j} {{x_j}}}}>0, \forall x_j$. Thus, the best response function of miner $i$ in~(\ref{Eq:BestResponse}) is always positive.

At last, we need to prove that $\lambda \mathscr F(x) > \mathscr F(\lambda x)$, for $\lambda >1$ for the scalability, and the steps of proof is shown in~(\ref{Eq:Scalability}). So far, we have proved that the best response function in~(\ref{Eq:BestResponse}) satisfies three properties of standard function described in~\cite{han2012game}. Therefore, the Nash equilibrium of MDG~$\mathcal{G}^u= \{\mathcal{N},\{x_i\}_{i \in \mathcal{N}},\{u_i\}_{i \in \mathcal{N}}\}$ is unique. The proof is completed.
\end{proof}
\textbf{Proof of Theorem 3:}
\begin{proof}[Sketch of proof]
According to~(\ref{Eq:FirstOrderUtility}), for each miner $i$, we have the mathematical expression $\frac{{\sum\nolimits_{i \ne j} {{x_j}} }}{{{{\left( {\sum\nolimits_{j\in \cal N}
 {{x_j}} } \right)}^2}}} = \frac{{p{e^{ - \lambda z{t_i}}}}}{{R + r{t_i}}}$. Then, we calculate the summation of this expression for all the miners as $\frac{{(N - 1)\sum\nolimits_{j \in \cal N} {{x_j}} }}{{{{\left( {\sum\nolimits_{j\in \cal N} {{x_j}} } \right)}^2}}} = \sum\nolimits_{i\in \cal N} {\frac{{p{e^{ - \lambda z{t_i}}}}}{{R + r{t_i}}}}$, which means $\frac{{(N - 1)}}{{\sum\nolimits_{j\in \cal N} {{x_j}} }} = \sum\nolimits_{i \in \cal N} {\frac{{p{e^{ - \lambda z{t_i}}}}}{{R + r{t_i}}}}$. Thus, we have $\sum\nolimits_{j \in \cal N} {{x_j}} = \frac{{N - 1}}{{\sum\nolimits_{i \in \cal N} {\frac{{p{e^{ - \lambda z{t_i}}}}}{{R + r{t_i}}}} }}$. Recall from~(\ref{Eq:BestResponse}), according to the first order derivative condition, we have $\sum\nolimits_{j \in \cal N} {{x_j}} = \sqrt {\frac{{(R + r{t_i})\sum\nolimits_{i \ne j} {{x_j}} }}{{p{e^{ - \lambda z{t_i}}}}}}$. With simple transformations, we obtain the Nash equilibrium for miner $i$ as shown in~(\ref{Eq:MDGequilibruim}). The detailed proof is given in~\cite{xiong2017blockchain} for the space limit.
\end{proof}
\bibliography{bibfile}

\end{document}